\documentclass[twoside,11pt]{article}

\usepackage{jmlr2e}

\def\jmlrheading {}

\jmlrheading{}
\ShortHeadings{Mutual Conditional Independence and Inference in Markov Networks}{Niharika}
\firstpageno{1}

\RequirePackage{amsmath,natbib}
\usepackage{tikz}
\usepackage[utf8]{inputenc}
\usepackage[upright]{fourier}
\usepackage{tkz-graph}
\usetikzlibrary{arrows}
\newcommand{\bigCI}{\mathrel{\text{\scalebox{1.07}{$\perp\mkern-10mu\perp$}}}}
\usepackage{graphicx}
\graphicspath{ {Images/} }
\usepackage{longtable}
\usepackage{booktabs}
\usepackage{array}

\begin{document}

\title{Mutual Conditional Independence and Its Applications to Inference in Markov Networks}
%Mutual Conditional Independence Among the Elements of an Independent Set and Its Applications to the Inference in Markov Networks}

\author{\name Niharika Gauraha  \email niharika.gauraha@gmail.com \\
       \addr Systems Science and Informatics Unit\\
       Indian Statistical Institute\\
       8th Mile, Mysore Road Bangalore, India
       }

\editor{}

\maketitle

\begin{abstract}%   <- trailing '%' for backward compatibility of .sty file
The fundamental concepts underlying in Markov networks are the conditional independence and the set of rules called Markov properties that translates conditional independence constraints into graphs. In this article we introduce the concept of mutual conditional independence relationship among elements of an independent set of a Markov network. We first prove that the mutual conditional independence property holds within the elements of a maximal independent set afterwards we prove equivalence between the set of  mutual conditional independence relations encoded by %rules implied by 
all the maximal independent sets and the three Markov properties(pair-wise, local and the global) under certain regularity conditions. The proof employs diverse methods involving graphoid axioms, factorization of the joint probability density functions and the graph theory. We present inference methods for decomposable and non-decomposable graphical models exploiting newly revealed mutual conditional independence property.
\end{abstract}


\section{Introduction}
A Markov network is a way of specifying conditional independence constraints between components of a multivariate distribution. Markov properties are the set of rules that determine how conditional independence constraints is translated into a graph. For details on Markov networks we refer the reader to \cite{LaurtBook} and \citet{Jordan}. The three Markov properties usually considered for Markov networks are pairwise, local and the global Markov properties. These Markov properties are equivalent to one another for positive distributions, for details on equivalence of Markov properties see \citet{Matus}. 

In an undirected graph, an independent set consists of mutually non-adjacent vertices or equivalently the elements of an independent set are mutually separated by the rest. For example let $G= \{ V,E \}$ be an undirected graph and let $I = \{V_1,V_2 , ..., V_k \}$ be an independent set of G then the vertices $\{V_1,V_2 , ..., V_k \}$ are mutually separated by $ \{ V \setminus I \}$. 

We extend the notion of similarity between separation in graph and conditional independence in probability to similarity between the mutual separation in graph and the mutual conditional independence in probability. The proof involves various methods from different disciplines; graphoid axioms, probability theory and the graph theory.  

Using graph theoretic concepts, we first prove that All the Maximal Independent Sets(AMIS) uniquely determine the graph. There is one-to-one relationship between graphs and AMIS. Then by applying probability theory concepts we prove that Mutual Conditional Independence Property(MCIP) holds among the elements of a maximal independent set. Since for any Markov network there will be a unique set of AMIS and hence a unique set of mutual conditional independence relations. Considering all mutual conditional independence relations obtained by AMIS, we derive an alternative formulation for the three Markov properties. Then we prove equivalence between the mutual conditional independence property and Markov properties, under positive distribution assumption.

Finally we shift our focus to the problem of probabilistic inference in Markov networks. In a multivariate set up, inference is the problem of computing a conditional probability distribution for the set of components where the values for some of the components are given, for details on inference on graphical models we refer to \citet{Wainwright}. 

We introduce inference methods that take the MCIP of the model into account. It provides quick answers to the queries by filtering on specific criteria. For example let G be any Markov network graph, let U represent the set of unobserved components and O represent the set of observed components or evidence and we wish to compute conditional probability of U given O. 
%\[
%	P(U | O) = \frac{P(U, O)}{P(O)}
%\]
%We introduce an inference algorithm to determine the association between the components of a Markov network.
%analyze interaction between components that takes the MCIP of the model into account. It provides quick answers to the queries by filtering on specific criteria. 
%Before applying any conventional inference algorithm, 
%Let us Suppose that U is a subset of any independent set say M and corresponding separator set $V \setminus M $ is a subset of E, Where V is a set of vertices. If above criteria holds then it is straight forward to conclude that given the set of evidence O the elements of U are mutually conditionally independent. We note that the time complexity for checking whether a set of vertices form independent set is linear in terms of number of vertices in the graph. 
The simplest possible inference is that suppose U forms an independent set and corresponding separator set $V \setminus U $ is same as set of evidence O, where V is a set of vertices then it is straight forward to conclude that given the set of evidence O the elements of U are mutually conditionally independent. We note that the time complexity for checking whether a set of vertices form an independent set is linear in terms of number of vertices in the graph. 
%Let us Suppose that U forms an independent and corresponding separator set $V \setminus U $ is same as set of evidence O, Where V is a set of vertices. If above criteria holds then it is straight forward to conclude that given the set of evidence O the elements of U are mutually conditionally independent. We note that the time complexity for checking whether a set of vertices form an independent set is linear in terms of number of vertices in the graph. 

Our approach of the inference based on MCIP will be very useful for  non-decomposable models where a closed form solution does not exist and for applications where it is more desirable to examine relationship among components than computing probabilities such as %graphical log-linear models 
analysis of categorical data and gene expression arrays. 
%We illustrate the applications of MCIP for inference in Gaussian Graphical Models also on simulated data.

The discussion below is organized as follows. In Section 2 we start with brief overview and mathematical foundations of the theory of Markov networks. Section 3 is concerned with proving that MCIP holds within the elements of an independent set. Section 4 involves deriving the global Markov property using the MCIP and proving equivalence between them. In section 5 we discuss some applications of mutual conditional independence relations in terms of statistical inferences for decomposable and non-decomposable graphical models. In section 6 we give computational details that we used for statistical inferences. In Section 7 we conclude and discuss future scope and applicability of MCIP. 

\section{Overview and Mathematical Foundations}
This section gives a general overview and mathematical foundations of Markov networks.
 
A graphical model is a technique for representation of the conditional independencies between variables in a multivariate probability distribution. The nodes or vertices in the graph correspond to random variables. The absence of an edge between two random variables denotes a conditional independence relation between them. In the literature several classes of graphs with various conditional independence interpretations have been described. Undirected graphical models (Markov Network) and directed acyclic graphs based graphical models(Bayesian Networks) are the most commonly known graphical models. In this article we only consider undirected graphical models, also known as Markov random fields or Markov networks. For details on the foundation of the theory of Markov networks we refer to \citet{LaurtBook} , \citet{Whittaker}, \citet{Preston} and\citet{Spitzer}.

\subsection{Graph Theory}
This section provides necessary concepts and definition of the graph theory that we will be using in later sections. For details on graph theory for graphical models we refer to \citet{LaurtBook}.
\subsubsection*{Notations}
A graph G, is a pair G = (V, E), where V is a set of vertices and E is a set of edges. 
\begin{definition}[Undirected Graphs]
A graph is  said to be an undirected graph if its vertices are connected by undirected edges. We consider only simple graph that has neither self loops nor multiple edges. \end{definition}
\begin{definition}[Maximal Independent set]
An independent set of a graph G is a subset S of nodes such that no two nodes in S are adjacent. An independent set is said to be maximal if no node can be added to S without violating independent set property.\end{definition}

\begin{theorem}[Uniqueness of Maximal Independent Sets]
	Given the complete list of maximal independent sets of a graph, the graph is uniquely determined.
\end{theorem}
\begin{proof}
	 Let $C_1, C_2,...C_m$ be the complete list of the maximal independent sets of nodes of a graph, then we have to show that the node set V is the union $  V= C_1 \cup C_2 \ \cdots \cup C_m $ and the edge set E consists of all unordered pairs $ \{ x, y \}$ of distinct elements of V such that $ \{ x, y \}$ is not contained in any of the $C_i$.

Clearly V contains the above union. Conversely, if $x \in V$ then $ \{ x \} $ is an independent set and hence is contained in some maximal independent set $C_i$ and then $x \in C_i$ and hence $x$ belongs to the union. This determines V. 

Also it is clear that any edge $ \{ x, y \}$  (where x, y are distinct elements of V) is not contained in any $C_i$. Conversely, if $ \{ x, y \}$ is a pair of nodes which is not an edge, then $ \{ x, y \} $ is an independent set and hence is contained in some $C_i$. Thus edge set E is also uniquely determined, as stated.
\end{proof}
\subsection{Conditional Independence}
In this section we define conditional independence in probability and Markov properties for Markov networks.%Graphical models are based on conditional independence restrictions for a set of variables with respect to all other variables. 
\begin{definition}[Conditional Independence] If $X, Y, Z $ are random variables with joint distribution P. Random variables X and Y are said to be conditionally independent given the random variable Z if following holds.%if the conditional distribution of X and Y given Z is the same as the product of conditional distribution of X given Z alone \citet{Dawid}, \citet{Whittaker} and \citet{Matus}.
\begin{align*}
	X  \bigCI Y \mid Z  &\iff P(X,Y \mid Z) = P(X \mid Z)P(Y \mid Z) \\
	& \iff p(X \mid Y,Z) = p(X \mid Z)\\	 
\end{align*}
\end{definition}

\subsection{Properties of Conditional Independence(graphoid axioms)}
We define some properties of conditional independence in terms of graphoid axioms as follows.
%, which we will make use of in proving Markov properties in later section.
%give reference Variants of them were first introduced by Dawid (1979) and further studied by Spohn (1980), Pearl and Paz (1985), Pearl (1988) and Geiger (1990). 
\begin{align}
	  &\textbf{Symmetry: } X \bigCI Y \mid Z \implies  Y \bigCI X \mid Z\\
	  \nonumber \\
	  &\textbf{Decomposition: } X \bigCI (Y \cup W) \mid Z \implies  X \bigCI Y \mid Z\\
	  \nonumber \\
	  &\textbf{Weak Union: } X \bigCI (Y \cup W) \mid Z \implies  X \bigCI Y \mid (Z\cup W)\\
	  \nonumber \\	  
	  &\textbf{Contraction: } X \bigCI (Y ) \mid Z \text{ and } X \bigCI (W) \mid (Z \cup Y) \implies  X \bigCI (Y \cup W) \mid (Z)\\	
	  \nonumber \\	  
	  &\textbf{Intersection: } X \bigCI (Y ) \mid (Z \cup W) \text{ and } X \bigCI (W) \mid (Z \cup Y) \implies  X \bigCI (Y \cup W) \mid (Z)	  
\end{align}
An alternative set of complete axioms we refer to \citet{Dan}.

\begin{definition}[semi-graphoid and graphoid]
A semi-graphoid is a dependency model which satisfies $Eq(1) - Eq(4)$. If also Eq$(5)$ holds it is called a graphoid.\end{definition}
\begin{definition}[ Probabilistic graphoid ]
In probability, Conditional independence defined as
\begin{align*}
	P(X, Y\mid Z) = P(X\mid Z)
\end{align*}\end{definition}
is a semi-graphoid. when P is strictly positive conditional independence becomes a graphoid.
\begin{definition}[ Graph Separation as graphoid]
Graph separation in undirected graph satisfies graphoid axioms. For details we refer to \citet{LaurtBook} and \citet{Dawid}.\end{definition}

\subsection{Markov Properties of Undirected Graphs}
In this sections we define the following three Markov properties for Markov networks. Let  $G= \{ V,E \}$ be an undirected graph and $P$ be a probability distribution over G.
\begin{definition}[(P) Pairwise Markov Property] 
The probability distribution P satisfies the pairwise Markov property for the graph G if for every pair of non adjacent edges X and Y , X is independent of Y given everything else. 
\begin{align*}
	 X \bigCI Y \mid (V \setminus {X,Y}) 
\end{align*}
\end{definition}
\begin{definition}[(L) Local Markov Property]
The probability distribution P satisfies the local Markov property for the graph G if every variable X is conditionally independent of its non-neighbours in the graph, given its neighbours. 
\begin{align*}
	X \bigCI (V \setminus {X \cup bd(X)}) \mid bd(X)
\end{align*}
where bd(X) denotes boundary of X.
\end{definition}
\begin{definition}[(G) Global Markov Property]
The probability distribution P, is said to be global Markov with respect to an undirected graph G if, for any disjoint subsets of nodes A, B, C such that C separates A and B on the graph, if and only if the distribution satisfies
\begin{align*}
	A \bigCI B \;| \;C
\end{align*}
\end{definition}

\begin{proposition}
Let G be an undirected graph. A probabilistic independence model that satisfies semi-graphoid axioms with respect to G, the following holds. For proof we refer \citet{LaurtBook}.
\begin{align*}
	(G) \implies (L) \implies (P)
\end{align*}\end{proposition}

\begin{proposition}
Let G be an undirected graph. A probabilistic independence model that satisfies graphoid axioms with respect to G, the following holds. For proof we refer \citet{Judea} and \citet{Dawid}.
\begin{align*}
	(G) \iff (L) \iff (P)
\end{align*}\end{proposition}

\subsection{Markov Network Graphs and Markov Network}
After discussing the graph theory, conditional independence and Markov properties for undirected graphs, now we are ready to define Markov network graphs and Markov networks.

\begin{definition}[Markov Network Graph]
A Markov network graph is an undirected graph G = ( V, E ) where $V= \{ X_1, X_2,..,X_n \}$ corresponds to random variables of a multivariate distribution. 
\end{definition}

\begin{definition}[Markov Network]
A Markov network is a tuple $M = (G, \psi)$ where G is a Markov network graph, $\psi = \{ \psi_1, \psi_2, ..., \psi_m\}$ is a set of non negative functions for each clique	$C_i \in G $ $\forall i = 1 \dots m $ and the joint pdf can be decomposed into factors as
\begin{align*}
	P(x) =\frac{1}{Z} \prod_{a \in C_m} \psi_a(x)
\end{align*}
where Z is a normalizing constant.
\end{definition}

\begin{theorem}[Hammersley-Clifford theorem]
Let  $M = (G, \psi)$ be a Markov network. Let probability density function $p(X)$ of the distribution of $X = \{ X_1, X_2,..,X_n \}$ be strictly positive. X satisfies global Markov property with respect to graph G if and only if it factorizes as follows.
\[
	P(x) = \prod_{a \in C_m} \psi_a(x)
\]
where $C_m$ are the maximal cliques of G and $\psi_a(x)$ depends on x through $ x_a = (x_v)_{v \in a}$ only.
\end{theorem}

It follows from the above discussion that if a strictly positive probability distribution factorizes with respect to G then it also satisfies all Markov properties(pair-wise,local and global) w.r.t. G.

\section{Mutual Conditional Independence }
In this section we prove that the elements of an independent set are mutually conditionally independent given the rest.
\begin{theorem}[Mutual Conditional Independence in Markov networks]
Let G be a Markov network graph and P(X) is a strictly positive probability which supports the conditional independences relations required to satisfy pairwise, local, and the global Markov property for G, then elements of an independent set I of G are mutually conditionally independent given the rest $\{ V \setminus I \}$.
\end{theorem}

\begin{proof}
Let $I = \{ X_1, X_2, ..., X_k \}$ be an independent set of G. Since $\{ X_1, X_2, ..., X_k \}$ are mutually non-adjacent, when we condition on $V \setminus I$ or equivalently when we remove the nodes $V \setminus I$ from G, the remaining vertices $\{ X_1, X_2, ..., X_k \}$ are disconnected which implies in probability complete independence among vertices of I.

Since $I = \{ X_1, X_2, ..., X_k \}$ form independent set they belong to separate cliques say $X_i \in C_i$ , for $ i = 1 \; to \; k$, where $C_i$ is a maximal clique in G. Without loss of generality we can assume that there are exactly k maximal cliques. From Theorem(Hammersley-Clifford theorem) the P factorizes as follows.
\[
	P = \psi_1(X_1, Y_1) \; \psi_2(X_2, Y_2) ... \psi_k(X_k, Y_k)\\	
\]
where $Y_i's$ are the sets of nodes that connects two or more $C_i's$ and each $ \{ X_i, Y_i \}$ forms a maximal clique in G. It can be noted that $Y_i$ can be empty in case of a disconnected graph and also union of $ \cup Y_i = V \setminus I$ .\\
The conditional probability $P(I \mid V \setminus I) $ can be expressed as 
\begin{align*}
	P(I \mid (Y_1 = y_1,..., Y_k =y_k) &= \psi_1(X_1, y_1) \; \psi_2(X_2, y_2) ... \psi_k(V_k, k_k)\\	
	P(I \mid V \setminus I) &= \phi_1(X_1) \; \phi_2(X_2) ... \phi_k(X_k)
\end{align*}
Hence $\{ X_1, X_2, ..., X_k \}$ are mutually conditionally independent given $\{ V \setminus I \}$.
\end{proof}

\section{Mutual Conditional Independence and the Markov Properties}
In this section we represent an alternative way to derive conditional independence relations required for satisfying the Markov properties of Markov networks. Specifically we prove equivalence between MCIP and pairwise Markov property and from proposition(12) it follows for other Markov properties(Local and the global). 
\begin{theorem}[Equivalence of MCIP and Markov properties]
Let G be a Markov network graph and let P be a strictly positive probability distribution which satisfies mutual conditional independence relations implied by maximal independent sets of G. Then conditional independence relations required to satisfy pairwise, local, and the global Markov properties for G also holds in P.
\end{theorem}

\begin{proof}
We prove equivalence of MCIP and pair-wise Markov property. Then under assumption of positive distribution MCIP is equal to local and the global Markov property as well.

Let $C_1, C_2,...C_m$ be the complete list of the maximal independent sets of nodes of G then as stated before $  V=  C_1 \cup C_2 \ \cdots \cup C_m $ and $E =  \{ ( x, y ) : ( x, y ) \notin C_i \; \forall i = 1 \cdots m \}$.
%then the node set V is the union $  V= C_1 \cup C_2 \ \cdots \cup C_m $ and the edge set E consists of all unordered pairs $ \{ x, y \}$ of distinct elements of V such that $ \{ x, y \}$ is not contained in any of the $C_i$. 

If conditional independence relations required to satisfy pairwise Markov property holds in P w.r.t. G, then elements of $C_i$ are mutually conditionally independent conditioned on $\{V \setminus C_i \}$ by Theorem 16. So we have proved as
\[
	\text{pairwise Marko Property } \implies MCIP
\]

We must recall that the Mutual conditional independence implies pair-wise conditional independence. Since elements of a $C_i$ are mutually conditionally independent given the $V \setminus C_i$, therefore they are also pairwise independent given $V \setminus C_i$. 

Now let us Suppose that $ \{ x, y \} \in V$ is a pair of nodes which is not an edge, then $ \{ x, y \} $ is an independent set and hence is contained in some $C_i$ and hence pairwise independent given the rest. Therefore
\[
	\text{pairwise Marko Property } \iff MCIP
\]
\end{proof}

We illustrate the proof by an example as follows. Let us consider the Markov network as given in figure (\ref{figure:1}).

%undirected graph construction
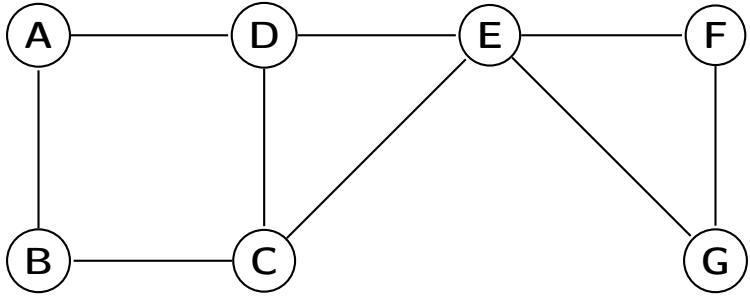
\begin{figure}
\begin{center}
\begin{tikzpicture}[>=stealth',shorten >=1pt,auto,node distance=3cm,
  thick,main node/.style={circle,draw,font=\sffamily\Large\bfseries}]
	
  \node[main node] (A) {A};
  \node[main node] (B) [below  of=A] {B};
  \node[main node] (D) [ right of=A] {D};
  \node[main node] (C) [below  of=D] {C};
  \node[main node] (E) [ right of=D] {E};
  \node[main node] (F) [ right of=E] {F};
  \node[main node] (G) [below of=F] {G};
    
  \path[every node/.style={font=\sffamily\small}]
   (A) edge node [left] {} (D)
      
   (B) edge node [right] {} (A)
   %edge node {} (D)
   (C) edge node [right] {} (B)      
   (D) edge node [left] {} (C)
   (D) edge node [right] {} (E)
	(C) edge node [right] {} (E)
	(E) edge node [right] {} (F)
	(E) edge node [right] {} (G)   
	(F) edge node [] {} (G);        
\end{tikzpicture}
\\
	\caption{A Markov network graph} \label{figure:1} 
\end{center}
\end{figure}
All Maximal independent sets for graph G are as :\\
\[
S = \{\{A,C,F\}, \{A,C,G\}, \{A,E\}, \{B,D,F\}, \{B,D,G\}, \{B,E\} \}
\]
Let us consider the first maximal independent sets $C_1 = \{A,C,F\}$ and let us suppose that MCIP holds which implies that A, C, F are mutually independent given rest of the random vectors {B,D,E,G}.\\
Or equivalently independence relation can be expressed as
\[
	A \bigCI C \bigCI F \mid (B,D,E,G) \\
\]
Applying weak union graphoid axiom (Equation (3)) to the above independence relation we get
\begin{align*}
	A \bigCI F \mid (B,D,E,G) \cup C  \\
	C \bigCI F \mid (B,D,E,G) \cup A \\	
	A \bigCI C \mid (B,D,E,G) \cup F \\
\end{align*}
Applying similarly arguments for the other set of maximal independent set we can show that for every non-adjacent pair $x,y \in V$ 
\[
	x \bigCI y \mid V \setminus \{ x,y \}
\] which is also a definition of pair-wise Markov property.

Conversely given pair-wise Markov property we have to show that MCIP holds.
From theorem (Hammersley-Clifford theorem), it is clear that under positive distribution assumption P satisfies pairwise Markov property with respect to graph G if and only if it factorizes as follows.
\[
	P(x) = \psi_1(A,B)\; \psi_2(A,D) \; \psi_3(B,C) \; \psi_4(C,D,E)\; \psi(E,F,G)
\]
Let us consider probability of (A,C,F) conditioned on (B,D,E,G), we obtain conditional probability as
\begin{align*}
	P(A,C,F \mid B = b,D=d,E=e,G=g) &= \phi_{11}(A,b)\; \phi_{12}(A,d) \; \phi_{21}(C,b) \;\phi_{22}(C,d,e) \; \phi_3(F,e,g)\\
	 &= \phi_1(A)\; \phi_2(C) \; \phi_3(F)
\end{align*}
From above factorization of pdf it follows that (A,C,F) are mutually independent conditioned on (B,D,E,G).\\
\\
Similarly we can show mutual conditional independence relations for the remaining maximal independent sets.
Hence it follows as 
\[
MCIP \iff \text{ pair-wise Markov property }
\]
Applying proposition 12 we get following equivalence relation that completes the proof.
\[
MCIP \iff P \iff L \iff G 
\]

\section{Applications and Illustrations}
In the following, we illustrate applications of MCIP for inference in Markov networks for discrete and continuous data set.
\\
%One possible application could be, when a set of observed variables corresponds to the mutual conditional independence(MCI) in P, statistical inference will be easier.\\
%For applications of conditional mutual independence based on cutsets of the graph, in the implication problem of probabilistic conditional independency and relational database we refer to \citet{Raymond}.

\subsection{Inference in Graphical Log-linear Models}
First we consider the discrete data set, the Reinis data taken from the "GRbase" R package(Risk factors for coronary heart disease, for details on Reinis dataset see \cite{Reinis}). 

\begin{example} [Decomposable Graphical Model for Rienis Dataset:]
The Reinis data is shown in the table (\ref{table:1}). 
\newpage
\begin{longtable}[h!] {@{}ccccc|cccccc@{}} 
		\caption{Reinis data }\label{table:1}\\
		\toprule \centering
		&&&&& Smoke & \multicolumn{2}{c}{no} &  \multicolumn{2}{c}{yes}\\
				%\cmidrule{6-7} \cmidrule{9-10}
		
				%\cmidrule{6-7} \cmidrule{9-10}
		%\toprule \centering
		Family& Protein & Systol & Phys  & & Mental & no & yes &&no& yes \\
		%\cmidrule{7-8}	
		\toprule \centering
		neg& $<3$ & $<140$ & no & & & 44 & 40 && 112 & 67\\
		&  &  & yes & & & 129 & 145 && 12 & 23 \\
		&  & $\geq140$ & no & &&  35 & 12 && 80 & 33\\
		&  &  & yes & & & 109 & 67 && 7 & 9 \\
		
		& $\geq3$ & $<140$ & no & & & 23 & 32 && 70 & 66 \\
		&  &  & yes & & & 50 & 80 && 7 & 13 \\
		&  & $\geq140$ & no & & & 24 & 25 && 73 & 57  \\
		&  &  & yes & & & 51 & 63 && 7 & 16 \\
		
		pos& $<3$ & $<140$ & no & &&  5 & 7 && 21 & 9 \\
		&  &  & yes & & & 9 & 17 && 1 & 4 \\
		&  & $\geq140$ & no & & & 4 & 3 && 11 & 8\\
		&  &  & yes & & &  14 & 17 && 5 & 2 \\
		
		& $\geq3$ & $<140$ & no & & & 7 & 3 && 14 & 14 \\
		&  &  & yes & & & 9 & 16 && 2 & 3 \\
		&  & $\geq140$ & no & &  & 4 &  0 && 13 & 11\\
		&  &  & yes & & &   5 & 14  && 4 &  4 \\
		\toprule \centering		
\end{longtable}

Using stepwise model selection for Reinis data, the best decomposable model we get is as given in figure (\ref{figure:2}) with following $\chi^2$ and $G^2$ test statistics. We use Wermuth's backward elimination algorithm, for details see \cite{Wermuth}.
\begin{align*}
	X^2  &= 51.11705 \\
	G^2 &= 51.35869 \\
	df & = 46\\
	X^2 << \chi^2(.95,46) &= 62.8, \text{ Hence the data supports the model  selected}.  
\end{align*}
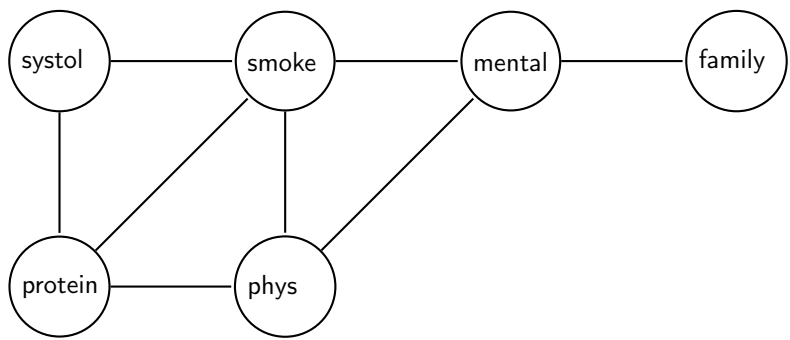
\begin{figure}[h!]  
\begin{center}
	\begin{tikzpicture}[>=stealth',shorten >=1pt,auto,node distance=3cm,
  thick,main node/.style={circle,text width=1cm, draw,font=\sffamily\small}]
	
  \node[main node] (1) {systol };
  \node[main node] (2) [below  of=1] {protein};
  \node[main node] (3) [ right of=1] {smoke };
  \node[main node] (4) [below  of=3] {phys   };
  \node[main node] (5) [right  of=3] {mental};
  \node[main node] (6) [right  of=5] {family};
 % \node[main node] (1) {family};
  
  \path[every node/.style={font=\sffamily\small}]
   (1) edge node [right] {} (2)      
   (2) edge node [right] {} (3)
   (1) edge node [left] {} (3)
   (3) edge node [right] {} (4)  
   (2) edge node [right] {} (4)
   (3) edge node [right] {} (5)
   (4) edge node [right] {} (5)
   (5) edge node [right] {} (6);
	\end{tikzpicture}
	\caption{A decomposable graphical model for the reinis data } \label{figure:2} 
\end{center}		
\end{figure}

We note that the variable set \{phys, systol,family\} forms an independent set as per the Markov network in figure (\ref{figure:2}). 

Now we derive a closed form expression for expected count using MCIP as follows. \\
\\
Let $X_1, X_2, X_3, X_4, X_5,X_6 $ represent the random variables Family, Protein, Systol, Phys, Smoke , Mental respectively.

\begin{align*}
	P(X_1, X_3,X_4 \mid X_2, X_5,X_6  ) &= \frac{P(X_1, X_2,X_3 ,X_4,X_5 ,X_6)}{P(X_2, X_5,X_6)}\\
	\text{ the joint pdf can be expressed as }&\\
	{P(X_1, X_2,X_3 ,X_4,X_5 ,X_6)} &= P(X_1, X_3,X_4 \mid X_2, X_5,X_6  ) * {P(X_2, X_5,X_6 )} \\	
	\\
	since \; X_1, X_3,X_4 & \text{ are mutually independent conditioned on } X_2, X_5,X_6 \\
	hence \; P(X_1, X_3,X_4 \mid X_2, X_5,X_6  ) & \; \text{ can be factorized as }\\	
	P(X_1, X_3,X_4 \mid X_2, X_5,X_6  ) &= P(X_1\mid X_2, X_5,X_6  )*P( X_3 \mid X_2, X_5,X_6  )*P(X_4 \mid X_2, X_5,X_6  )\\	
\end{align*}
{ the joint pdf can be written as }
\[
	{P(X_1, X_2,X_3 ,X_4,X_5 ,X_6)} = P(X_1\mid X_2, X_5,X_6  )*  P( X_3 \mid X_2, X_5,X_6  )*P(X_4 \mid X_2, X_5,X_6  ) * {P(X_2, X_5,X_6 )}
	\]
	\[
	=  \frac{(X_1, X_2, X_5,X_6  )*P(X_3, X_2, X_5,X_6  )* P(X_4, X_2, X_5,X_6  )} {P(X_1, X_2, X_5,X_6)^2}\]
	%expected count can be computed as $\hat{m}_{ijklmn} = \hat{p}_{ijklmn}* N$, where $N = \sum_{ijklmn} n_{ijklmn}$ \\
	%substituting value of $\hat{p}_{ijklmn} $ we get following closed form expression for the expected cell counts.\\
	After simplification we get following closed form expression for the maximum likelihood estimator of the expected cell counts. For details on computing closed form expressions for the expected cell counts for decomposable log-linear models, see \cite{BishopBook}.
	\begin{align*}
		\hat{m_{ijklmn}} &= \frac{n_{ij..mn}  n_{.jk.mn}  n_{.j.lmn}}{n_{.j..mn}^2}\\
		\text{ where } n_{ijklmn} &= \text{ observed count in cell }(i,j,k,l,m,n)\\
	\hat{m_{ijklmn}} &= \text{ Maximum Likelihood Estimate of the expected cell count }  m_{ijklmn}\\
	\end{align*}
	Under the mutual conditional independence assumption the table of fitted values is given in the table(\ref{table:2}).	
%\newpage	
	
	\begin{longtable}[h!] {@{}ccccc|cccccc@{}} 
		\caption{Fitted values for Reinis table }\label{table:2}\\
		\toprule \centering
		&&&&& Smoke & \multicolumn{2}{c}{no} &  \multicolumn{2}{c}{yes}\\
				%\cmidrule{6-7} \cmidrule{9-10}
		
				%\cmidrule{6-7} \cmidrule{9-10}
		%\toprule \centering
		Family& Protein & Systol & Phys  & & Mental & no&yes &  &no&yes \\
		%\cmidrule{7-8}	
		\toprule \centering
		neg& $<3$ & $<140$ & no & & & 42.828483  & 4.323380  && 37.102750 &   3.745388\\
		&  &  & yes & & & 127.025386 & 12.822752 && 110.043382 &11.108480 \\
		&  & $\geq140$ & no & &&  336.061224 &  6.010204 && 17.081633   &2.846939\\
		&  &  & yes & & & 143.081633&  23.846939 && 67.775510  &11.295918 \\
		
		& $\geq3$ & $<140$ & no & & & 111.297302  &20.044064 && 78.517959 & 14.140675 \\
		&  &  & yes & & & 12.421574 &  2.237061  && 8.763165 &  1.578200 \\
		&  & $\geq140$ & no & & & 66.211530  &11.536857 && 33.427180 &  5.824433 \\
		&  &  & yes & & &  21.504599  & 3.747014 && 10.856691  & 1.891696 \\
		
		pos& $<3$ & $<140$ & no & &&   25.526279 &  4.311871 && 24.092218 & 4.069631 \\
		&  &  & yes & & & 50.612449 &  8.549400 && 47.769053 &  8.069097 \\
		&  & $\geq140$ & no& & & 28.956142 &  4.777763 &&  22.546004 &  3.720091\\
		&  &  & yes & & &  83.490210 & 13.775885 && 65.007644 & 10.726261 \\
		
		& $\geq3$ & $<140$ & no & & & 68.758172 & 14.452355 && 71.715512 & 15.073961\\
		&  &  & yes & & & 8.089197  & 1.700277  && 8.437119 &  1.773407 \\
		&  & $\geq140$ & no & &  & 63.788280 & 13.429112&&  58.472590 & 12.310019\\
		&  &  & yes & & &   15.516068  & 3.266541&&  14.223062 &  2.994329 \\
		\toprule \centering		
	\end{longtable}
	
	To test if the above model holds, we perform Perason's chi-square test. We use the follwing formula.
	\[
		 X^2 = \sum_i \frac{(O_i-E_i)^2}{E_i}
	\]
	where O denotes observed cell count and E as expected cell count. \\
	The following test statistcs we get 
\begin{align*}  
	X^2  &= 35.01022 \\
	df & = 46\\
	\chi^2(.95,46) &=  62.8   
\end{align*}
As per the Chi-Squared test, the data supports the mutual conditional independence among  \{phys, systol,family\} conditioned on  \{phys, systol,family\}.
%We performed the mutual independence check for \{phys, systol,family\} conditioned on \{smoke,mental ,protein\} , the result is as
%\[
%	X^2 = 35.01022 << \chi_{.05}^{2} = 62.8 \; with \; D.F. = 46 
%\]
For details on graphical log-linear model we refer the reader to \cite{Christensen}, and \cite{BishopBook}.
\end{example}

\begin{example} [Non-Decomposable Graphical Model for Rienis Dataset:]
Let us consider the Reinis data once again. We get the best non-decomposable graphical model as given in the figure (\ref{figure:3}) with following $\chi^2$ and $G^2$ test statistics.
\begin{align*}
	X^2  &= 61.87653 \\
	G^2 &= 62.84262 \\
	df & = 49\\
	X^2 < \chi^2(.95,49) &= 66.3, \text{ Hence the data supports the model  selected }.  
\end{align*}

We notice that in this model the following mutual conditional relation holds.
\[
	  ( phys \bigCI systol \bigCI family ) \mid  (phys, systol,family).
\]
Since the above relation is same as relation we got for decomposable model, hence the factorization of joint pdf, getting closed form expression for expected cell counts, chi-square test for model testing are the same as decomposable model as computed previously in example 1.
\begin{figure}[h!]  
\begin{center}
	\begin{tikzpicture}[>=stealth',shorten >=1pt,auto,node distance=3cm,
  thick,main node/.style={circle,text width=1cm, draw,font=\sffamily\small}]
	
  \node[main node] (1) {systol };
  \node[main node] (2) [right  of=1] {protein};
  \node[main node] (3) [ below of=2] {smoke };
  \node[main node] (4) [right  of=2] {mental   };
  \node[main node] (5) [below  of=4] {phys};
  \node[main node] (6) [right  of=5] {family};
 % \node[main node] (1) {family};
  
  \path[every node/.style={font=\sffamily\small}]
   (1) edge node [right] {} (2)      
   (2) edge node [right] {} (3)
   (1) edge node [left] {} (3) 
   (2) edge node [right] {} (4)
   (3) edge node [right] {} (5)
   (4) edge node [right] {} (5)
   (5) edge node [right] {} (6);
	\end{tikzpicture}
	\caption{A non-decomposable graphical model for the reinis data } \label{figure:3} 
\end{center}		
\end{figure}
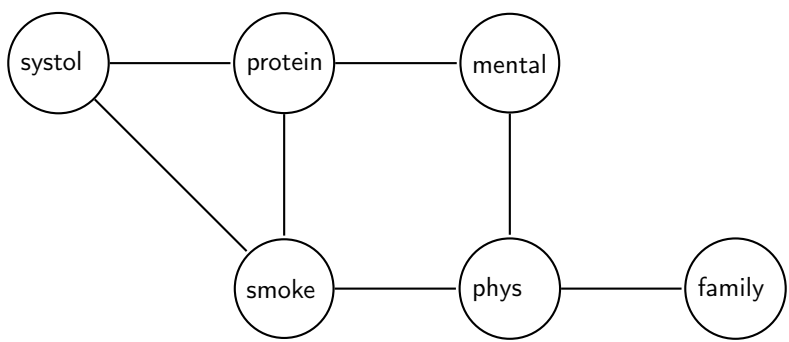
We also note that the closed form expression we get for estimated cell counts for non-decomposable graphical model must be same as one of the decomposable models. Hence MCIP can be directly used to get closed form estimates for non-decomposable graphical models without converting it to decomposable graphical models.
%since Independent set and clique are complementary notions
%We must note that to get closed form expression for computing expected cell counts for non-decomposable graphical model, MCIP acn be used.
\end{example}

\subsection{Inference in Gaussian Graphical Models}
In this section, we illustrate application of MCIP for inference in Gaussiam graphical models(GGM). We consider the "seeds"  dataset which is available at UCI machine learning repository  \\
https://archive.ics.uci.edu/ml/datasets/seeds. For details on GGM we refer to some selected classical and recent research papers \cite{Dempster}, \cite{mohan14a}, \cite{tan14b} and \cite{janzamin14a}
\begin{example}[A Decomposable Model for the Seeds dataset]
The best decomposable graphical model we get using stepwise selection method is given in figure(\ref{figure:4}).
\begin{figure}[h!]  
\begin{center}
	\begin{tikzpicture}
  \SetVertexNormal[Shape      = circle,
                   FillColor  = white,
                   LineWidth  = 1pt]
  \SetUpEdge[lw         = 1pt,
             color      = black,
             labelcolor = white,
             labeltext  = red,
             labelstyle = {sloped,draw,text=blue}]
 %\tikzstyle{EdgeStyle}=[bend left]
 \Vertex[L = $V_4$, x=0, y=0]{A} 
 \Vertex[L=$V_2$,x=0, y=3]{B}
 \Vertex[L=$V_1$, x=3, y=5]{C}
 \Vertex[L=$V_5$, x=4, y=2]{D}
 \Vertex[L=$V_3$,x=8, y=3]{E}
 \Vertex[L=$V_6$, x=7, y=0]{F}
 \Vertex[L=$V_7$, x=3, y=-1]{G}
 \Edges(A,B,D,A) \Edges(D,G,A) \Edges(B,G) \Edges(B,C) \Edges(B,E) \Edges(B,F) \Edges(C,E) \Edges(C,G) \Edges(C,D) \Edges(E,D) \Edges(E,G) \Edges(E,F) \Edges(F,G)
% \Edges(G,A,P,Q,E) \Edges(C,A,Q) \Edges(C,R,G) \Edges(P,E,A)
	\end{tikzpicture}
	\caption{A decomposable graphical model for the seeds dataset } \label{figure:4} 
\end{center}		
\end{figure}
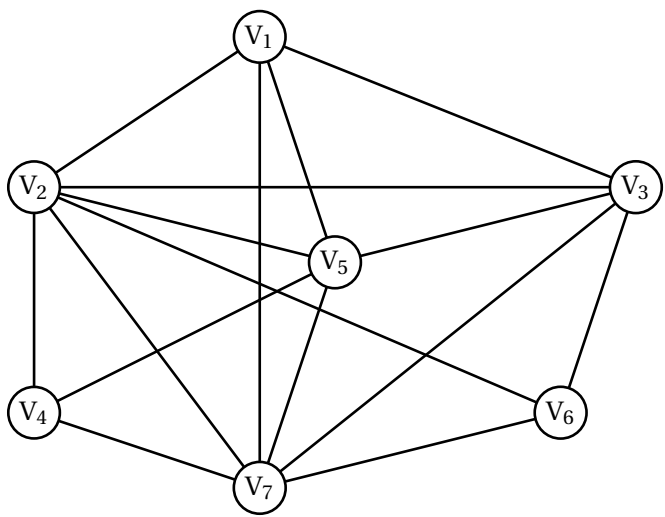
\end{example}

As per the graph in figure (\ref{figure:4}), the vertex set $\{ V_1,V_4,V_6 \}$ forms an independent set. The following conditional test results also supports that the variables $\{ V_1,V_4,V_6 \}$  are pair wise conditionally independent conditioned on the rest $\{ V_2,V_3,V_5, V_7 \}$.\\
\\
Test:$ V1 \bigCI V6 \mid V2,\; V3,\; V5,\; V7 $\\
Chi-Square test statistic :   $ 0.893 $ df:$ 1$ p-value: $0.3447 $\\
\\
Test $V1 \bigCI V4 \mid V2,\; V3,\; V5,\; V7 $\\
Chi-Square test statistic:   $ 1.055$ df: $1$ p-value: $0.3044$ \\
\\
Test $V4 \bigCI V6 \mid V2,\; V3,\; V5,\; V7 $\\
Chi-Square test statistic:    $0.952$ df: $1$ p-value: $0.3293 $\\
\\
For normal variables, zero correlation implies independence and pairwise independence implies mutual independence. Hence the variables $ V1 , V4 , V6 $ are also mutually conditionally independent conditioned on $\{ V_2,V_3,V_5, V_7 \}$. Equivalently it can be expressed as
\[ V1 \bigCI V4 \bigCI V6 \mid V2,\; V3,\; V5,\; V7 
\]

\section{Computational details}
All the experimental results in this paper were carried out using R 3.1.3 with the packages gRim and MASS. All packages used are 
available at http://CRAN.R-project.org/.

\section{Conclusion}
In summary, we discussed different Markov properties for the class of Markov networks. We derived an alternative formulation of the Markov properties of Markov networks. We have given a new perspective on conditional independence over an independent set as mutual conditional independence. We have proved equivalence between MCIP and the Markov properties, under positive distribution assumption. We have presented MCIP based approach for inference. The experimental results are carried out for the proposed MCIP based approach for inferences on discrete and continuous datasets. MCIP can be a promising new direction for model selection and inference in Markov networks.
%We also proved that for any undirected graph and any probability distribution which satisfies the mutual conditional independence with respect to the maximal independent set of the graph will also respect the pairwise, local and the global Markov properties of the graph.
%Applications of MCI with respect to maximal independent sets may further investigated for directed and chain graphs based graphical models.

% http://jmlr.csail.mit.edu/papers/volume1/heckerman00a/heckerman00a.pdf


\newpage

\vskip 0.2in
\bibliography{niharika15a}

\end{document}